\documentclass[12pt]{article}

\usepackage{amsmath,amssymb,amsthm}
\usepackage{microtype}
\usepackage[margin=1.15in]{geometry}
\usepackage{tikz}
\usetikzlibrary{patterns}
\usepackage{pgfplots}
\pgfplotsset{compat=1.17}
\usepackage{booktabs}
\usepackage{graphicx}
\usepackage{float}
\usepackage[colorlinks=true,linkcolor=blue,citecolor=blue,urlcolor=blue]{hyperref}

\newtheorem{theorem}{Theorem}
\newtheorem{lemma}{Lemma}
\newtheorem{proposition}{Proposition}
\theoremstyle{remark}

\begin{document}

\title{Shape of Yield Curves\\ under Instrument-Based Parallel Shifts}
\author{Jian Sun \\ jiansun@uchicago.edu \\FTI  Consulting }
\date{}
\maketitle

\begin{abstract}
This paper studies the shape of zero and par yield curves through the carry
and convexity of zero-bond and par-swap butterflies under
\emph{instrument-based parallel shifts}: the component
instruments are held fixed, and each instrument's own quote receives the
same additive shock. Although this framework is a conditional scenario
experiment rather than a model-independent restriction on admissible
yield-curve shapes, it is simple and widely used in fixed-income risk
management. For zero-coupon bonds, maturity is the natural horizontal
coordinate; for par swaps, the nonlinear bootstrap relation makes the
fixed-leg annuity, or PV01, the relevant coordinate. Under the stated
instrument-based evolution rule, an appropriately weighted zero-cost
butterfly has a non-negative instantaneous revaluation under either direction of an
admissible common shift, while convexity of
the initial curve in the relevant coordinate also generates positive
short-horizon carry. These results provide a partial economic rationale
for the commonly observed concavity of yield curves: convex segments can
support positive-carry, level-neutral butterflies under the benchmark
shift convention.

We then examine daily U.S. dollar par swap rates from
July 3, 2000 through October 28, 2016. The data are obtained from
the Federal Reserve Bank of St. Louis FRED database. The empirical results show that convexity is more
frequently detected in annuity coordinates than under maturity
weighting and is concentrated primarily at the front end of the
curve. Realized short-horizon movements are seldom
instrument-parallel, and the gross advantage of convex butterflies
is substantially reduced by non-level curve movements and transaction
costs. The evidence therefore supports the use of annuity coordinates
for constructing and evaluating swap butterflies, while also showing
that the benchmark parallel-shift scenario should be interpreted as a
risk-management device rather than as a description of typical realized
curve dynamics.
\end{abstract}

\noindent\textbf{Keywords:} yield curve; zero-coupon bond; swap par curve; butterfly; annuity factor; PV01; carry; convexity; transaction costs.

\noindent\textbf{JEL Classification:} G12, G13, E43.

\section{Introduction}

In fixed-income markets, the yield curve is among the most closely monitored indicators for traders, economists, and risk managers.  A yield curve plots interest rates across contract maturities from the short end to the long end.  For a given currency, the curve describes the relationship between the level of interest rates, or the cost of borrowing, and time to maturity.  U.S. Treasury yields at different maturities form the Treasury curve, while fixed rates on interest-rate swaps across tenors form the swap curve.

The shape of the yield curve is commonly interpreted as containing information about future interest-rate changes and economic activity.  The most common qualitative descriptions are normal, inverted, flat, and humped curves.  A normal curve is upward sloping: longer-maturity bonds have higher yields than shorter-maturity bonds.  An inverted curve has higher short-term yields than long-term yields.  A flat curve has similar yields across maturities, while a humped curve rises over some maturities and falls over others.

\begin{figure}[htbp]
\centering
\begin{tikzpicture}
\begin{axis}[
    width=0.88\textwidth,
    height=0.44\textwidth,
    title={USD par and zero curves on 2016-10-28},
    title style={font=\small},
    xlabel={Maturity (years)},
    ylabel={Rate (\%)},
    xlabel style={font=\small},
    ylabel style={font=\small},
    tick label style={font=\small},
    xmin=0.5,
    xmax=31,
    ymin=0.9,
    ymax=2.2,
    xtick={1,2,3,4,5,7,10,30},
    ytick={1.0,1.2,1.4,1.6,1.8,2.0,2.2},
    grid=both,
    major grid style={gray!22},
    minor grid style={gray!12},
    axis lines=left,
    axis line style={->},
    enlargelimits=false,
    line cap=round,
    line join=round,
    legend style={
        at={(0.03,0.97)},
        anchor=north west,
        draw=none,
        fill=white,
        font=\small
    }
]

% The intermediate points are interpolation guides only.
\addplot[
    blue,
    thick,
    smooth,
    tension=0.35,
    mark=none
] coordinates {
    (1,1.000)
    (2,1.100)
    (3,1.190)
    (4,1.280)
    (5,1.360)
    (6,1.440)
    (7,1.520)
    (8,1.583)
    (9,1.647)
    (10,1.710)
    (15,1.800)
    (20,1.890)
    (25,1.980)
    (30,2.070)
};
\addlegendentry{Par swap rate}

\addplot[
    red!70!black,
    thick,
    dashed,
    smooth,
    tension=0.35,
    mark=none
] coordinates {
    (1,0.998)
    (2,1.098)
    (3,1.188)
    (4,1.279)
    (5,1.360)
    (6,1.442)
    (7,1.525)
    (8,1.590)
    (9,1.656)
    (10,1.723)
    (15,1.816)
    (20,1.915)
    (25,2.020)
    (30,2.133)
};
\addlegendentry{Bootstrapped zero rate}

% Markers appear only at the observed FRED par-rate maturities.
\addplot[
    only marks,
    blue,
    mark=*,
    mark size=2.4pt,
    forget plot
] coordinates {
    (1,1.000)
    (2,1.100)
    (3,1.190)
    (4,1.280)
    (5,1.360)
    (7,1.520)
    (10,1.710)
    (30,2.070)
};

\addplot[
    only marks,
    red!70!black,
    mark=square*,
    mark size=2.1pt,
    forget plot
] coordinates {
    (1,0.998)
    (2,1.098)
    (3,1.188)
    (4,1.279)
    (5,1.360)
    (7,1.525)
    (10,1.723)
    (30,2.133)
};

\end{axis}
\end{tikzpicture}
\caption{U.S. dollar par swap rates and bootstrapped continuously compounded zero rates on October 28, 2016. Markers denote quoted maturities; connecting lines are visual guides.}
\label{fig:usd-par-zero-curve-2016-10-28}
\end{figure}

Figure~\ref{fig:usd-par-zero-curve-2016-10-28} shows the observed U.S. dollar par swap curve and its implied zero curve on the final date in the sample \cite{fredswap}.  Both curves are upward sloping and broadly concave over the available maturity range, while their separation at longer maturities illustrates that par and zero rates are distinct quote coordinates linked through the bootstrap.  The purpose of the paper is not to claim that convex curve segments are ruled out by model-independent no-arbitrage.  Instead, it asks how the initial shape of a curve enters the carry of a butterfly under a specified family of common yield shocks and which coordinate should be used to construct the corresponding trade.

Different instruments give rise to different, but closely related versions of the yield curve.  Zero-coupon bond yields define the zero curve.  Coupon bonds trading at par define the par curve.  Standard fixed-for-floating swaps are also par instruments: the fixed swap rate is the coupon that makes the fixed leg have the same value as the floating leg at inception.  In an idealized setting, discount factors, zero rates, forward rates, and par rates determine each other through no-arbitrage relations.

Since the period covered by our data, the U.S. dollar interest-rate
market has transitioned from LIBOR to the Secured Overnight Financing
Rate (SOFR). The Alternative Reference Rates Committee selected SOFR as
its preferred alternative to U.S. dollar LIBOR, and publication of
U.S. dollar LIBOR swap-rate benchmarks subsequently ceased
\cite{arrcsofr,icesofrswap}. Modern U.S. dollar discount curves are
therefore commonly constructed from SOFR overnight-indexed swaps rather
than from LIBOR swaps.

This institutional change does not alter the central curve-construction
principle used in this paper. A quoted par SOFR overnight-indexed swap
rate is still determined by equality of the present values of its fixed
and floating legs, and its fixed-leg annuity remains the relevant PV01
coordinate. Consequently, discount factors can again be recovered
sequentially from par swap quotations, subject to the applicable
payment, day-count, overnight-compounding, interpolation, and collateral
conventions. The bootstrap inputs and cash-flow details have changed,
but the nonlinear relation between par rates, discount factors, and
fixed-leg annuities---which is the relation used in our theoretical
analysis---remains.

When inspecting the plotted zero and par curves, two natural questions
arise. First, can either curve continue to rise indefinitely as maturity
increases, and what restrictions apply to its long-maturity behavior?
Second, why are yield curves so often broadly concave, particularly at the
long end? Standard no-arbitrage arguments provide only a partial answer.
In a deterministic term-structure setting, positive discount factors that
decrease with maturity rule out elementary static arbitrage, but they
remain consistent with a wide variety of zero- and par-curve shapes.
Concavity is therefore not a model-independent no-arbitrage requirement.
This paper approaches these questions through zero-bond and par-swap
butterfly positions. By studying their carry and revaluation under
instrument-based common shifts, we show how convex segments can support
positive-carry, level-neutral butterflies and thereby provide a partial
economic rationale for the prevalence of concave yield curves. The
analysis also clarifies how the long-maturity behavior of par rates is
linked to the limiting behavior of fixed-leg annuities and discount
factors.

The zero-coupon analysis separates two effects.  First, a zero-cost, duration-neutral butterfly has positive convexity with respect to an instantaneous common change in the three quoted yields; this component is present independently of the initial curve shape.  Second, under the instrument-based parallel-shift convention specified below, convexity of the three selected initial yield points produces positive benchmark holding-period carry.  The latter conclusion is conditional on that evolution rule.  It is not a statement that the position remains profitable under arbitrary roll-down, slope, or curvature movements.

The main theoretical contribution concerns par and swap curves.  A par rate at a given maturity depends on the entire discount-factor curve up to that maturity, so a maturity-weighted argument cannot be transferred directly from zero-coupon bonds.  We identify the fixed-leg annuity, or PV01, as the natural horizontal coordinate.  Under instrument-based parallel shifts of the selected par rates, the bootstrap relation implies two-sided monotonicity and curvature properties for shifted annuities.  These properties determine the appropriate annuity-weighted par-swap butterfly and show how convexity of the initial par curve in annuity coordinates enters its carry.

The coordinate principle also has a direct practitioner analogue in futures-based curve butterflies.  Calendar-maturity weights alone generally do not produce a level-neutral futures position because contract-specific DV01s, conversion factors, and, for Treasury futures, cheapest-to-deliver characteristics differ across contracts.  Desks therefore scale wing positions by DV01 or PV01-type exposures rather than by maturity alone.  The annuity-weighted swap butterfly developed below is the swap-market version of this exposure-balancing principle; it does not assert that the additional delivery and basis risks of futures contracts are identical to those of swaps.

The instrument-based parallel-shift assumptions used in the paper are deliberately stylized and are not intended to constitute a complete dynamic model of the term structure.  Actual curve movements contain level, slope, curvature, roll-down, and basis components.  The results should therefore be read as conditional dominance and trade-design diagnostics: a convex segment identifies a butterfly with favorable benchmark carry and non-negative revaluation within the specified shift family, while the position remains exposed to non-level curve movements.

A common parallel bump nevertheless remains a central risk-management experiment.  Aggregate DV01 or PV01 is typically computed by increasing the relevant rates or market quotes by the same small amount and revaluing the portfolio; bucketed or key-rate measures are then used to diagnose departures from that level direction.  It is therefore meaningful to study yield-curve shape, butterfly carry, and conditional shape restrictions along the parallel-shift direction even though that direction does not span all realized curve risk.  In this sense, the paper analyzes the geometry associated with the standard total-rates-risk bump and then measures empirically how nonparallel changes alter the result.

The empirical analysis makes this distinction explicit.  We construct daily single-curve discount factors and fixed-leg annuities from public U.S. dollar swap rates at maturities from one to thirty years.  We first compare annuity-coordinate and maturity-coordinate classifications across all tenor triplets.  We then follow every annuity-convex butterfly to the next available curve, infer the own-instrument par-rate change of each original swap after roll-down, and calculate the realized mark-to-market.  This produces a direct test of the instrument-parallel assumption and quantifies the role of transaction costs.  The data show that annuity weighting materially changes curve-shape classification, especially at the front and intermediate parts of the curve, but that daily changes are predominantly nonparallel and the small gross P\&L advantage does not survive a representative bid--ask calculation.

\section{Related Literature}

There is an extensive literature on yield-curve modeling.  Nelson and Siegel \cite{nelson} proposed a parsimonious parametric model for the forward curve of the form
\begin{equation*}
  f(t)=\beta_{1}+\beta_{2}e^{-t/\tau}+\beta_{3}\frac{t}{\tau}e^{-t/\tau},
\end{equation*}
where $\tau$ is a nonlinear parameter and $\beta_{1},\beta_{2},\beta_{3}$ are linear parameters.  Svensson \cite{svensson} and later extensions add additional hump factors to improve flexibility.  Arbitrage-free versions of Nelson--Siegel-type specifications were developed by Christensen, Diebold, and Rudebusch \cite{christensen2009,christensen2011}.

A related practical literature concerns the construction of discount,
zero, and forward curves from observed market instruments.  Zero rates
and discount factors are generally not quoted directly at every maturity;
instead, they are recovered by bootstrapping deposits, bonds, futures,
forward-rate agreements, and swaps.  The procedure uses the sequential
cash-flow structure of these instruments to solve recursively for discount
factors, supplemented by an interpolation convention between quoted
maturities.  Standard treatments of fixed-income curve construction and
bootstrapping are given, for example, by Andersen and Piterbarg
\cite{andersenpiterbarg2010} and Fabozzi \cite{fabozzi2007}.

For a spot-starting par swap with fixed-leg payment dates
$t_1,\ldots,t_n$, accrual fractions $\alpha_j$, and discount factors
$p_j=P(0,t_j)$, the par rate $x_n$ satisfies
\begin{equation*}
    x_n A_n = 1-p_n,
    \qquad
    A_n=\sum_{j=1}^{n}\alpha_jp_j.
\end{equation*}
Thus, the par curve and the zero curve are distinct but linked through the
bootstrap: a par rate depends not only on the terminal discount factor but
also on the entire fixed-leg annuity.  The implied shape of the zero and
forward curves between quoted maturities also depends on the chosen
interpolation and calibration conventions.  This nonlinear relation is
central to the present paper because it explains why maturity is the
natural coordinate for zero-bond butterflies, whereas the fixed-leg
annuity, or PV01, is the natural coordinate for par-swap butterflies.

Bootstrapping constructs the initial term structure from observed market
quotes.  Dynamic term-structure models address the separate question of
how that curve evolves over time.  Vasicek \cite{vasicek1977} modeled the short rate by an Ornstein--Uhlenbeck process,
\begin{equation*}
  dr_t=\lambda(\theta-r_t)dt+\sigma dW_t,
\end{equation*}
while Cox, Ingersoll, and Ross \cite{cir1985} proposed the square-root model
\begin{equation*}
  dr_t=\kappa(\theta-r_t)dt+\sigma\sqrt{r_t}dW_t.
\end{equation*}
Ho and Lee \cite{ho1986} and Hull and White \cite{hull1990b} developed arbitrage-free models calibrated to an initial curve.  Heath, Jarrow, and Morton \cite{heath1992} characterized arbitrage-free forward-rate dynamics.  These models describe the stochastic evolution or parametric representation of the term structure.  The present paper focuses on a different question: how does the initial shape of a curve enter the carry and factor exposure of a butterfly under a specified instrument-based parallel-shift experiment?

Classical fixed-income risk analysis emphasizes duration, convexity, and the dependence of immunization results on the assumed family of yield changes.  Ingersoll, Skelton, and Weil \cite{ingersoll1978} give a foundational discussion of duration and the limitations of single-factor interest-rate risk measures.  The present paper is in that tradition: its dominance statements are explicitly conditional on the selected instrument-shift experiment.

Other related work connects yield-curve shape with volatility and common factors.  Litterman and Scheinkman \cite{litterman1991} showed that yield-curve movements are largely explained by level, slope, and curvature factors.  Litterman, Scheinkman, and Weiss \cite{littermanweiss1991} documented a relationship between implied volatility and yield-curve butterfly spreads.  Brown and Schaefer \cite{brown1994,brown2000} studied the relation between interest-rate volatility and the shape of long-term forward rates.  The present paper is complementary: rather than deriving curve shape from a fully specified stochastic model, it develops a direct trading decomposition under stylized instrument-based parallel shifts and identifies the annuity coordinate relevant for par-swap butterflies.  Adrian, Crump, and Moench \cite{adrian2013} provide a modern regression-based term-structure framework with multiple yield factors.  Bj\"ork and Christensen \cite{bjorkchristensen1999} and Filipovi\'c \cite{filipovic1999} study consistency between parametric curve families and arbitrage-free interest-rate dynamics.  Those consistency questions reinforce the distinction made here between a static quote-coordinate experiment and a complete dynamic term-structure model.

\section{Rates, Discount Factors, and Par Rates}

We work in a simplified annual-tenor setting first.  Let the current time be $0$, and let future payment dates be $1,2,\ldots,n$.  Denote by $p_i$ the discount factor to year $i$, and by $y_i$ the continuously compounded zero rate to year $i$.  Then
\begin{equation}\label{eq:zero-discount}
  p_i=e^{-i y_i}.
\end{equation}
The annual-payment convention is used only to simplify notation.  In market applications, the par curve is normally inferred from fixed-for-floating interest-rate swaps, whose fixed legs may pay semiannually, quarterly, or according to another market convention.  With accrual factors $\alpha_i$, the annuity factor is $A_n=\sum_{i=1}^{n}\alpha_i p_i$ and the par rate is $x_n=(1-p_n)/A_n$.  The annual formulation isolates the relevant curve geometry; weighted-annuity and continuous-payment analogues are discussed below, and the empirical implementation uses semiannual accruals.

The one-period forward rate for the interval $(i,i+1)$ is
\begin{equation*}
  f_i=\frac{p_i}{p_{i+1}}-1=e^{(i+1)y_{i+1}-i y_i}-1.
\end{equation*}
Elementary cash-and-carry no-arbitrage requires positive discount factors; under non-negative one-period forward rates, the discount factors are also decreasing in maturity.

The annual par rate $x_i$ for maturity $i$ is determined by the condition that the fixed coupon leg plus final principal has value one:
\begin{equation}\label{eq:par-rate}
  x_i(p_1+\cdots+p_i)+p_i=1.
\end{equation}
Equivalently,
\begin{equation}\label{eq:par-rate-formula}
  x_i=\frac{1-p_i}{P_i},\qquad P_i=p_1+\cdots+p_i,
\end{equation}
where $P_i$ is the annuity factor, or fixed-leg PV01 up to the conventional notional and one-basis-point scaling.  Conversely, if the par rates $x_i$ are given, the discount factors are obtained recursively by
\begin{equation}\label{eq:bootstrap-discrete}
  p_i=\frac{1-x_i(p_1+\cdots+p_{i-1})}{1+x_i}.
\end{equation}
Thus the zero curve and par curve contain equivalent information, subject to the usual positivity and monotonicity restrictions on discount factors.

\section{Zero-Coupon Bond Butterflies}

We first consider three zero-coupon bonds with maturities $T_1<T_2<T_3$ and continuously compounded zero yields $y_1,y_2,y_3$.  We construct a butterfly portfolio that buys dollar amounts $\lambda_1$ and $\lambda_3$ of the short- and long-maturity zero-coupon bonds and shorts a dollar amount $\lambda_2$ of the intermediate bond.  The weights are chosen to have zero initial cost and zero first-order exposure to an instantaneous common change in the three quoted yields:
\begin{equation*}
  \lambda_1+\lambda_3=\lambda_2,
  \qquad
  \lambda_1T_1+\lambda_3T_3=\lambda_2T_2.
\end{equation*}
Up to a positive scalar,
\begin{equation}\label{eq:zero-weights}
  \lambda_1=T_3-T_2,
  \qquad
  \lambda_2=T_3-T_1,
  \qquad
  \lambda_3=T_2-T_1.
\end{equation}
Hence $\lambda_1/\lambda_2+\lambda_3/\lambda_2=1$, and $T_2$ is the weighted average of $T_1$ and $T_3$ with weights $\lambda_1/\lambda_2$ and $\lambda_3/\lambda_2$.

\begin{proposition}[Instantaneous common-yield-shift convexity]\label{prop:instant-zero}
Suppose that, at the same instant and with maturities held fixed, the three quoted zero yields change from $y_i$ to $y_i+a$.  The value of the initially zero-cost butterfly in \eqref{eq:zero-weights} is
\begin{equation*}
  V^{\mathrm{shift}}(a)
  =\lambda_1e^{-aT_1}+\lambda_3e^{-aT_3}-\lambda_2e^{-aT_2}\geq 0.
\end{equation*}
The inequality is strict when $a\neq0$ and the maturities are distinct.
\end{proposition}

\begin{proof}
By Jensen's inequality applied to the convex function $u\mapsto e^{-au}$,
\begin{equation*}
  \frac{\lambda_1}{\lambda_2}e^{-aT_1}
  +\frac{\lambda_3}{\lambda_2}e^{-aT_3}
  \geq
  e^{-a\left(\frac{\lambda_1}{\lambda_2}T_1+\frac{\lambda_3}{\lambda_2}T_3\right)}
  =e^{-aT_2}.
\end{equation*}
Multiplying by $\lambda_2$ gives the result.
\end{proof}

The proposition describes the convex response of a duration-neutral portfolio to an instantaneous common change in the three yields.  It does not depend on the initial values $y_1,y_2,y_3$ and therefore does not, by itself, impose a restriction on the initial shape of the zero curve.

We next state the holding-period convention used in the original curve-shape argument.  Let $0<h<T_1$, fix the three zero-coupon bonds selected at time $0$, and follow those same securities over the holding period.  For the bond initially maturing at $T_i$, assume that its time-$h$ yield at remaining maturity $T_i-h$ is
\begin{equation}\label{eq:instrument-based-parallel-shift}
  y_h(T_i-h)=y_0(T_i)+a=y_i+a,
\end{equation}
where the additive shock $a$ is common to the three bonds.  We refer to \eqref{eq:instrument-based-parallel-shift} as an \emph{instrument-based parallel shift}.  The shift is parallel across the fixed set of instruments being monitored: each bond's yield at the end of the holding period equals its own initial yield plus the same additive shock.

This convention is natural for desk-level scenario analysis, where a trader follows a fixed butterfly and shocks the quotes of its component instruments.  An analogous distinction arises in volatility markets, where sticky-strike and sticky-delta conventions describe the evolution of a fixed set of vanilla-option quotes or quote coordinates rather than a complete dynamic law for the entire future volatility surface.  Likewise, the instrument-based parallel shift is an economically useful scenario rule for the selected bond set, not a claim about the future zero curve at every remaining maturity.

A dollar initially invested in bond $i$ is then worth
\begin{equation*}
  \frac{e^{-(T_i-h)(y_i+a)}}{e^{-T_i y_i}}
  =e^{h y_i-a(T_i-h)}
\end{equation*}
at time $h$.

\begin{theorem}[Holding-period dominance under an instrument-based parallel shift]\label{thm:zero-instrument-shift}
Assume that the three selected initial yields are convex in maturity in the sense that
\begin{equation}\label{eq:zero-convexity}
  (T_3-T_2)y_1+(T_2-T_1)y_3\geq(T_3-T_1)y_2.
\end{equation}
Under the instrument-based parallel-shift convention \eqref{eq:instrument-based-parallel-shift}, the value of the initially zero-cost butterfly is
\begin{equation}\label{eq:zero-instrument-shift-value}
  V(a,h)
  =\lambda_1e^{h y_1-a(T_1-h)}
  +\lambda_3e^{h y_3-a(T_3-h)}
  -\lambda_2e^{h y_2-a(T_2-h)}
  \geq0
\end{equation}
for every common shock $a$.  If \eqref{eq:zero-convexity} is strict and $h>0$, the inequality in \eqref{eq:zero-instrument-shift-value} is strict.
\end{theorem}

\begin{proof}
Write
\begin{equation*}
  w_1=\frac{\lambda_1}{\lambda_2},
  \qquad
  w_3=\frac{\lambda_3}{\lambda_2}.
\end{equation*}
Then $w_1+w_3=1$ and $w_1T_1+w_3T_3=T_2$.  Using Jensen's inequality,
\begin{align*}
&w_1e^{h y_1-a(T_1-h)}+w_3e^{h y_3-a(T_3-h)}\\
&\qquad\geq
e^{h(w_1y_1+w_3y_3)-a\left[w_1(T_1-h)+w_3(T_3-h)\right]}\\
&\qquad=
e^{h(w_1y_1+w_3y_3)-a(T_2-h)}.
\end{align*}
Condition \eqref{eq:zero-convexity} is equivalent to $w_1y_1+w_3y_3\geq y_2$.  Hence the last expression is at least $e^{h y_2-a(T_2-h)}$.  Multiplying by $\lambda_2$ proves the result.
\end{proof}

The word ``instrument-based'' is essential.  The convention \eqref{eq:instrument-based-parallel-shift} differs from a parallel shift of the yield curve indexed by remaining maturity.  The latter would impose
\begin{equation}\label{eq:rolled-parallel-shift}
  y_h(\tau)=y_0(\tau)+a,
\end{equation}
so that the bond initially maturing at $T_i$ would be repriced using $y_0(T_i-h)+a$, rather than $y_0(T_i)+a$.  Under \eqref{eq:rolled-parallel-shift}, convexity of only the three initial points $(T_i,y_i)$ does not determine the sign of the holding-period butterfly return.  The theorem therefore concerns a fixed-instrument trading scenario, not an arbitrary deformation of the rolled term structure.

For a short holding period, the two scenario conventions are nevertheless close when the initial curve is smooth.  If $y_0(T)$ is twice differentiable, then
\begin{equation}\label{eq:short-horizon-shift-closeness}
  y_0(T_i)-y_0(T_i-h)
  =h\,\partial_T y_0(T_i)+O(h^2).
\end{equation}
Thus the difference between the instrument-based quote $y_0(T_i)+a$ and the rolled-curve quote $y_0(T_i-h)+a$ is a roll-down term of order $h$.  Over a one-day horizon this distinction is normally small relative to the level of the yield, provided the curve is smooth and its local slope is moderate.  This statement concerns the proximity of two \emph{scenario definitions}; it does not claim that observed one-day market changes are themselves parallel.  The instrument-based convention remains important because it is explicitly relative to the existing position: the same component securities or swaps are followed through time, and their own quotes are the objects being shocked.

More generally, define
\begin{equation*}
  \delta_i(h)=y_h(T_i-h)-y_i
\end{equation*}
and let $(s_1,s_2,s_3)=(1,-1,1)$.  Under an arbitrary future curve, the value of the initial dollar butterfly can be decomposed exactly as
\begin{align}
  V_h
  & =\sum_{i=1}^{3}s_i\lambda_i e^{h y_i}e^{-\delta_i(h)(T_i-h)}\nonumber\\
  & =\underbrace{\sum_{i=1}^{3}s_i\lambda_i e^{h y_i}}_{\text{zero-shock carry benchmark}}
  +\underbrace{\sum_{i=1}^{3}s_i\lambda_i e^{h y_i}
  \left[e^{-\delta_i(h)(T_i-h)}-1\right]}_{\text{roll-down and yield-change component}}.
  \label{eq:zero-pnl-decomposition}
\end{align}
For a short holding period, the benchmark component is
\begin{equation*}
  h\big(\lambda_1y_1+\lambda_3y_3-\lambda_2y_2\big)+O(h^2),
\end{equation*}
which is positive under strict convexity.  The second component in \eqref{eq:zero-pnl-decomposition} captures the slope, curvature, roll-down, and other yield changes that can offset this benchmark carry.  Theorem~\ref{thm:zero-instrument-shift} is therefore a conditional dominance result within the instrument-based parallel-shift family, not a model-independent prohibition of convex zero-curve segments.

\section{Par and Swap Curve Butterflies in Annuity Coordinates}

We now turn to par rates.  The par curve is more difficult than the zero curve because the par rate at maturity $n$ depends on the entire discount-factor curve up to $n$.  This nonlinear bootstrap effect makes the annuity factor $P_n$, rather than calendar maturity, the natural horizontal coordinate for constructing the butterfly.

\subsection{Instantaneous common-rate-shift convexity}

For annual payments, the par-rate relation is
\begin{equation*}
  x_n=\frac{1-p_n}{P_n},\qquad P_n=p_1+\cdots+p_n.
\end{equation*}
Equivalently, this is the par swap rate for a fixed leg with unit annual accruals.  In the real swap market, one replaces $P_n$ by the fixed-leg annuity $A_n=\sum_{i=1}^{n}\alpha_i p_i$, where $\alpha_i$ are the year fractions determined by the relevant day-count and payment-frequency convention.  The proof uses only positivity of accrual weights and the annuity-coordinate ordering, so the annual-payment assumption is a normalization rather than a restriction.
We assume throughout this section that discount factors are positive and decreasing, and that the par curve is increasing:
\begin{equation*}
  x_1<x_2<\cdots.
\end{equation*}

\begin{theorem}[Long-end restriction]\label{thm:long-end}
Assume that $p_0=1$ and that the discount factors $p_n$ are positive and decreasing.  Define
\begin{equation*}
  P_n=\sum_{i=1}^{n}p_i,
  \qquad
  x_n=\frac{1-p_n}{P_n}.
\end{equation*}
Then the par-rate sequence has a finite limit.  More precisely, exactly one of the following alternatives holds:
\begin{enumerate}
  \item If $P_n\to\infty$, then $x_n\to0$.
  \item If $P_n\to P_\infty<\infty$, then $p_n\to0$ and
  \begin{equation*}
    x_n\longrightarrow x_\infty=\frac{1}{P_\infty} < +\infty.
  \end{equation*}
\end{enumerate}
Consequently, if the par-rate sequence is positive and increasing, then the second alternative must hold: the increasing par curve converges to the finite positive limit $x_\infty=1/P_\infty$.
\end{theorem}

\begin{proof}
Because $p_0=1$ and the discount factors are positive and decreasing,
\begin{equation*}
  0<p_n\leq1,
  \qquad
  0\leq1-p_n\leq1.
\end{equation*}
If $P_n\to\infty$, then
\begin{equation*}
  0\leq x_n=\frac{1-p_n}{P_n}\leq\frac{1}{P_n}\longrightarrow0.
\end{equation*}
If instead $P_n\to P_\infty<\infty$, the positive series $\sum_{i=1}^{\infty}p_i$ converges, so $p_n\to0$.  Hence
\begin{equation*}
  x_n=\frac{1-p_n}{P_n}\longrightarrow\frac{1}{P_\infty}>0.
\end{equation*}
Finally, if $x_n$ is increasing and $x_1>0$, it cannot converge to zero.  The first alternative is therefore impossible, and the second must hold.
\end{proof}

Although Theorem~\ref{thm:long-end} is elementary, it gives a clear structural interpretation of the long end of an increasing positive par-swap curve.  Such a curve cannot continue to rise without bound.  The numerator $1-p_n$ in the par-rate formula is bounded by one, while the fixed-leg annuity $P_n$ accumulates positive discounted cash flows as maturity increases.  If that annuity were to diverge, the par rate would converge to zero, which is incompatible with an increasing curve that starts at a positive rate.  Therefore the annuity must approach a finite limit, the terminal discount factor must vanish, and the long-term par rate approaches the reciprocal of the limiting annuity:
\begin{equation*}
  x_\infty=\frac{1}{P_\infty}.
\end{equation*}
Economically, increasingly distant fixed-leg payments make progressively smaller contributions to present value, so the quoted swap rate eventually stabilizes rather than growing indefinitely.  This asymptotic flattening is consistent with the long-end behavior commonly observed in market swap curves.  The conclusion is an asymptotic implication of positive decreasing discount factors together with an increasing positive par curve; it does not require the curve to be monotone over every finite maturity interval.

Now fix the quoted par instruments and apply an instrument-based parallel shift $y$, so that the new par rate is $x_n+y$.  Let $p_n(y)$ and $P_n(y)$ denote the resulting bootstrapped discount factors and annuity factors.  We call $y$ \emph{admissible} if the shifted bootstrap produces positive discount factors for every maturity under consideration.  When an argument uses all intermediate shifts between $0$ and $y$, admissibility is understood to hold throughout that interval.  The bootstrap recursion is
\begin{equation}\label{eq:shifted-annuity-recursion}
  P_n(y)=\frac{1+P_{n-1}(y)}{1+x_n+y},
  \qquad P_0(y)=0.
\end{equation}

\begin{lemma}[Annuity monotonicity]\label{lem:annuity-mono}
For every admissible shift $y$ and every $n\geq1$,
\begin{equation*}
  \operatorname{sgn}\big(P_n(y)-P_n\big)=-\operatorname{sgn}(y).
\end{equation*}
Thus $P_n(y)<P_n$ when $y>0$, while $P_n(y)>P_n$ when $y<0$.
\end{lemma}

\begin{proof}
For $y>0$,
\begin{equation*}
  P_1(y)=\frac{1}{1+x_1+y}<\frac{1}{1+x_1}=P_1.
\end{equation*}
If $P_{n-1}(y)<P_{n-1}$, then \eqref{eq:shifted-annuity-recursion} gives
\begin{equation*}
  P_n(y)=\frac{1+P_{n-1}(y)}{1+x_n+y}
  <\frac{1+P_{n-1}}{1+x_n}=P_n.
\end{equation*}
The argument for $y<0$ is identical with the inequalities reversed.
\end{proof}

\begin{lemma}[Two-sided par-rate shift inequalities]\label{lem:parallel-shift-ineq}
For every admissible shift $y$ and every $n\geq1$,
\begin{align}
  \min\{y,0\}
  &\leq \frac{1}{P_n(y)}-\frac{1}{P_n}
  \leq \max\{y,0\},\label{eq:two-sided-inverse-annuity}\\
  \min\{y,0\}
  &\leq \frac{p_n}{P_n}-\frac{p_n(y)}{P_n(y)}
  \leq \max\{y,0\}.\label{eq:two-sided-discount-annuity}
\end{align}
In particular, both differences in \eqref{eq:two-sided-inverse-annuity} and \eqref{eq:two-sided-discount-annuity} have the same sign as $y$.
\end{lemma}

\begin{proof}
Set
\begin{equation*}
  Q_n(y)=\frac{1}{P_n(y)}.
\end{equation*}
From \eqref{eq:shifted-annuity-recursion},
\begin{equation}\label{eq:Q-recursion}
  Q_n(y)=\frac{(1+x_n+y)Q_{n-1}(y)}{1+Q_{n-1}(y)},
  \qquad n\geq2,
\end{equation}
with $Q_1(y)=1+x_1+y$.  Differentiating \eqref{eq:Q-recursion} yields
\begin{equation}\label{eq:Q-derivative}
  Q_n'(y)
  =\frac{Q_{n-1}(y)}{1+Q_{n-1}(y)}
  +\frac{(1+x_n+y)Q_{n-1}'(y)}{\big(1+Q_{n-1}(y)\big)^2}.
\end{equation}
Admissibility implies $p_n(y)>0$, which is equivalent to
\begin{equation*}
  x_n+y<Q_{n-1}(y).
\end{equation*}
Since $Q_1'(y)=1$, induction in \eqref{eq:Q-derivative} gives
\begin{equation}\label{eq:Q-derivative-bounds}
  0\leq Q_n'(y)\leq1
\end{equation}
throughout the admissible interval.  Integrating \eqref{eq:Q-derivative-bounds} between $0$ and $y$ proves \eqref{eq:two-sided-inverse-annuity}.

Using
\begin{equation*}
  x_n=\frac{1}{P_n}-\frac{p_n}{P_n},
  \qquad
  x_n+y=\frac{1}{P_n(y)}-\frac{p_n(y)}{P_n(y)},
\end{equation*}
we obtain
\begin{equation}\label{eq:shift-difference-decomposition}
  \left(\frac{p_n}{P_n}-\frac{p_n(y)}{P_n(y)}\right)
  +\left(\frac{1}{P_n(y)}-\frac{1}{P_n}\right)=y.
\end{equation}
Combining \eqref{eq:shift-difference-decomposition} with \eqref{eq:two-sided-inverse-annuity} proves \eqref{eq:two-sided-discount-annuity}.
\end{proof}

\begin{lemma}[Two-sided discount-factor ratio monotonicity]\label{lem:ratio-mono}
Assume that the par curve is increasing.  For an admissible positive shift $y>0$, the ratio
\begin{equation*}
  \frac{p_n(y)}{p_n}
\end{equation*}
is decreasing in $n$.  For an admissible negative shift $y<0$, the ratio is increasing in $n$.
\end{lemma}

\begin{proof}
Let $d_n=x_n-x_{n-1}\geq0$.  From $p_n=1-x_nP_n$ and its shifted counterpart,
\begin{align*}
  p_{n-1}-p_n
  &=d_nP_n+x_{n-1}p_n,\\
  p_{n-1}(y)-p_n(y)
  &=d_nP_n(y)+(x_{n-1}+y)p_n(y).
\end{align*}
A direct rearrangement gives
\begin{align}\label{eq:ratio-monotonicity-identity}
  &p_{n-1}(y)p_n-p_n(y)p_{n-1}\nonumber\\
  &\quad=d_nP_nP_n(y)
  \left(\frac{p_n}{P_n}-\frac{p_n(y)}{P_n(y)}\right)
  +y\,p_np_n(y).
\end{align}
By Lemma~\ref{lem:parallel-shift-ineq}, the bracketed difference has the same sign as $y$.  Every other factor in \eqref{eq:ratio-monotonicity-identity} is non-negative.  Hence the left-hand side has the same sign as $y$.  Dividing by $p_{n-1}p_n>0$ gives
\begin{equation*}
  \frac{p_{n-1}(y)}{p_{n-1}}
  \geq \frac{p_n(y)}{p_n}
  \quad\text{when }y>0,
\end{equation*}
and the reverse inequality when $y<0$.
\end{proof}

\begin{lemma}[Curvature of shifted annuities]\label{lem:annuity-curvature}
Let $n<m<k$.  If $p_i(y)/p_i$ is decreasing in $i$, then the three points
\begin{equation*}
  (P_n,P_n(y)),\qquad(P_m,P_m(y)),\qquad(P_k,P_k(y))
\end{equation*}
are concave.  If $p_i(y)/p_i$ is increasing in $i$, then the three points are convex.  Consequently, the shifted-annuity graph is concave for $y>0$ and convex for $y<0$ when the initial par curve is increasing.
\end{lemma}

\begin{proof}
Concavity is equivalent to
\begin{equation}\label{eq:discrete-annuity-concavity}
  (P_k-P_n)P_m(y)
  \geq
  (P_m-P_n)P_k(y)+(P_k-P_m)P_n(y).
\end{equation}
After cancellation, \eqref{eq:discrete-annuity-concavity} becomes
\begin{equation}\label{eq:double-sum-annuity}
  \left(\sum_{i=m+1}^{k}p_i\right)
  \left(\sum_{j=n+1}^{m}p_j(y)\right)
  \geq
  \left(\sum_{i=m+1}^{k}p_i(y)\right)
  \left(\sum_{j=n+1}^{m}p_j\right).
\end{equation}
If $p_i(y)/p_i$ is decreasing, then for every $i>j$,
\begin{equation*}
  p_i p_j(y)\geq p_i(y)p_j,
\end{equation*}
so \eqref{eq:double-sum-annuity} follows term by term.  If the ratio is increasing, every inequality reverses and the three points are convex.  The final statement follows from Lemma~\ref{lem:ratio-mono}.
\end{proof}

\begin{theorem}[Shift component of an annuity-weighted swap butterfly]\label{thm:swap-shift-component}
Assume that the par curve is increasing and that $y$ is an admissible instrument-based parallel shift.  Select three maturities with annuities $P_1<P_2<P_3$ and choose relative notionals
\begin{equation}\label{eq:swap-butterfly-weights}
  \lambda_1=P_3-P_2,
  \qquad
  \lambda_2=P_3-P_1,
  \qquad
  \lambda_3=P_2-P_1.
\end{equation}
Consider a receive-fixed position in the short- and long-maturity swaps and a pay-fixed position in the intermediate swap, each initially struck at its par rate.  Immediately after the common shift $x_i\mapsto x_i+y$, the butterfly value is
\begin{equation}\label{eq:swap-shift-component}
  V_{\mathrm{shift}}(y)
  =-y\left[
  \lambda_1P_1(y)+\lambda_3P_3(y)-\lambda_2P_2(y)
  \right].
\end{equation}
Then
\begin{equation*}
  V_{\mathrm{shift}}(y)\geq0
\end{equation*}
for every admissible positive or negative shift $y$.
\end{theorem}

\begin{proof}
A unit-notional receive-fixed swap struck at its initial par rate $x_i$ has value
\begin{equation*}
  \big(x_i-(x_i+y)\big)P_i(y)=-yP_i(y)
\end{equation*}
immediately after the shift.  Summing the three positions gives \eqref{eq:swap-shift-component}.

For $y>0$, Lemma~\ref{lem:annuity-curvature} implies concavity of the shifted-annuity graph.  With the weights in \eqref{eq:swap-butterfly-weights}, this is equivalent to
\begin{equation*}
  \lambda_1P_1(y)+\lambda_3P_3(y)-\lambda_2P_2(y)\leq0.
\end{equation*}
Since $-y<0$, \eqref{eq:swap-shift-component} is non-negative.  For $y<0$, the shifted-annuity graph is convex, so the bracketed expression is non-negative; since $-y>0$, the product is again non-negative.
\end{proof}

The theorem isolates the pure shift component of the par-swap butterfly.  Unlike the zero-coupon case, where non-negative common-shift revaluation follows directly from Jensen's inequality, the swap case requires the preceding lemmas.  Those lemmas transfer a common par-rate shift through the nonlinear bootstrap relation into a curvature statement for shifted annuities: the map $P_i\mapsto P_i(y)$ is concave after a positive shift and convex after a negative shift.  This transformed annuity geometry gives the sign of the bracket in \eqref{eq:swap-shift-component}; the factor $-y$ then makes the instantaneous shift component non-negative in both directions.

The sign of this instantaneous component does not depend on convexity of the initial par curve.  Initial curve shape enters through the short-horizon carry of the fixed-rate positions.  We now derive that component in monetary units.

\subsection{Short-horizon carry of the swap butterfly}

Consider a short holding period $h$ that ends before the first contractual payment date of all three swaps.  Let $\delta_h>0$ denote the accrued day-count fraction over the holding period.  Assume that the three spot-starting swaps share the same first floating-rate accrual period and floating index, as is standard when neighboring maturities are selected from the same market swap curve.  Let $r_h$ denote the floating rate accrued over this common first period up to time $h$.

The accrued fixed-minus-floating amount of a unit-notional receive-fixed swap with fixed rate $x_i$ is
\begin{equation*}
  \delta_h(x_i-r_h).
\end{equation*}
For the receive-fixed short and long positions and the pay-fixed intermediate position, the accrued coupon component of the butterfly's dirty value is therefore
\begin{align}
  C_h
  & =\delta_h\left[
  \lambda_1(x_1-r_h)+\lambda_3(x_3-r_h)-\lambda_2(x_2-r_h)
  \right]\nonumber\\
  & =\delta_h\left(
  \lambda_1x_1+\lambda_3x_3-\lambda_2x_2
  \right),
  \label{eq:swap-short-carry}
\end{align}
because the weights in \eqref{eq:swap-butterfly-weights} satisfy
\begin{equation*}
  \lambda_1+\lambda_3=\lambda_2.
\end{equation*}
Thus the common floating accrual cancels exactly.  If the accrued amount at time $h$ is expressed as a time-zero present value, it is multiplied by the discount factor $D(0,h)$:
\begin{equation}\label{eq:swap-short-carry-pv}
  \operatorname{PV}_0(C_h)
  =D(0,h)\,\delta_h\left(
  \lambda_1x_1+\lambda_3x_3-\lambda_2x_2
  \right).
\end{equation}

\begin{theorem}[Short-horizon carry in annuity coordinates]\label{thm:swap-short-carry}
Select three par swaps and choose the relative notionals in \eqref{eq:swap-butterfly-weights}.  If the three initial par-curve points $(P_1,x_1)$, $(P_2,x_2)$, and $(P_3,x_3)$ are convex in the annuity coordinate, namely
\begin{equation}\label{eq:par-annuity-convexity}
  \lambda_1x_1+\lambda_3x_3-\lambda_2x_2\geq0,
\end{equation}
then the short-horizon accrued coupon carry in \eqref{eq:swap-short-carry} and its time-zero present value in \eqref{eq:swap-short-carry-pv} are non-negative.  They are strictly positive when the annuity-coordinate convexity inequality is strict and $h>0$.
\end{theorem}

\begin{proof}
Equation \eqref{eq:swap-short-carry} shows that the accrued carry is the positive accrual factor $\delta_h$ multiplied by the annuity-coordinate convexity spread in \eqref{eq:par-annuity-convexity}.  Discounting by $D(0,h)>0$ preserves the sign.
\end{proof}

To separate this accrued carry from the remaining mark-to-market effects, define
\begin{equation*}
  L_x=\lambda_1x_1+\lambda_3x_3-\lambda_2x_2.
\end{equation*}
If $V_h^{\mathrm{dirty}}$ denotes the total dirty value of the butterfly at time $h$, then
\begin{equation}\label{eq:swap-carry-revaluation-decomposition}
  V_h^{\mathrm{dirty}}
  =
  \underbrace{\delta_h L_x}_{\text{accrued coupon carry}}
  +
  \underbrace{\left(V_h^{\mathrm{dirty}}-\delta_h L_x\right)}_{
    \substack{\text{roll-down, annuity changes, and}\\
              \text{curve-revaluation component}}}.
\end{equation}
Theorem~\ref{thm:swap-short-carry} determines the sign of the first component.  The second component is not signed by the carry argument: it includes passage of time through the remaining cash-flow schedules, changes in the annuities, and level and non-level movements of the market curve.  Under the instantaneous instrument-based common-shift experiment, the pure shift component is non-negative by Theorem~\ref{thm:swap-shift-component}; over an actual holding period, however, roll-down and nonparallel movements can offset or dominate the accrued carry.

Theorem~\ref{thm:swap-shift-component} and \ref{thm:swap-short-carry} give the same componentwise structure as in the zero-coupon construction.  Convexity of the initial zero curve in maturity coordinates, or of the initial par curve in annuity coordinates, produces positive short-horizon butterfly carry.  The corresponding instrument-based parallel-shift component is non-negative.  In the swap case, the first statement follows from cancellation of the common floating accrual, while the second is the consequence of the annuity monotonicity and curvature lemmas.

These two statements concern identifiable components of the position.  The total realized holding-period profit and loss can also contain roll-down, changes in the remaining annuities, and non-level movements of the curve.  Accordingly, annuity-coordinate convexity identifies a butterfly with positive benchmark carry and non-negative revaluation under instrument-based level shifts; it does not eliminate exposure to slope, curvature, or other departures from the assumed common-shift experiment.

\subsection{Continuous par curve}

The continuous-payment version clarifies the role of annuity factors and can be viewed as the limiting case of the discrete swap argument as the payment mesh tends to zero.  Let $p(T)$ be the discount factor to maturity $T$, and define the continuous annuity factor
\begin{equation*}
  P(T)=\int_0^T p(s)\,ds.
\end{equation*}
The continuously paid par rate $x(T)$ is determined by
\begin{equation}\label{eq:cont-par}
  x(T)P(T)+p(T)=1,
\end{equation}
so
\begin{equation*}
  x(T)=\frac{1-p(T)}{P(T)}.
\end{equation*}
Since $P'(T)=p(T)$, equation \eqref{eq:cont-par} is equivalent to
\begin{equation}\label{eq:annuity-ode}
  P'(T)+x(T)P(T)=1,
  \qquad P(0)=0.
\end{equation}
If the par curve undergoes an instrument-based parallel shift by $y$, the shifted annuity factor $P(T,y)$ solves
\begin{equation}\label{eq:shifted-annuity-ode}
  \frac{\partial P(T,y)}{\partial T}+(x(T)+y)P(T,y)=1,
  \qquad P(0,y)=0.
\end{equation}

\begin{theorem}[Continuous shifted-annuity curvature]\label{thm:continuous-par}
Let $x(T)$ be an increasing continuously paid par curve, and suppose that $x(T)+y$ is admissible over the relevant maturity range.  Let $P(T)$ and $P(T,y)$ solve \eqref{eq:annuity-ode} and \eqref{eq:shifted-annuity-ode}.  Then
\begin{equation*}
  P(T,y)<P(T)\quad\text{for }y>0,
  \qquad
  P(T,y)>P(T)\quad\text{for }y<0.
\end{equation*}
Moreover, for any $0<T_1<T_2<T_3$, the graph
\begin{equation*}
  P(T)\longmapsto P(T,y)
\end{equation*}
is concave when $y>0$ and convex when $y<0$.  Consequently, with
\begin{equation*}
  \lambda_1=P(T_3)-P(T_2),\qquad
  \lambda_2=P(T_3)-P(T_1),\qquad
  \lambda_3=P(T_2)-P(T_1),
\end{equation*}
the instantaneous shift component of the continuously paid receive-fixed butterfly is non-negative for every admissible positive or negative shift:
\begin{equation}\label{eq:continuous-shift-component}
  -y\left[
  \lambda_1P(T_1,y)+\lambda_3P(T_3,y)-\lambda_2P(T_2,y)
  \right]\geq0.
\end{equation}
\end{theorem}

\begin{proof}
Subtracting \eqref{eq:shifted-annuity-ode} from \eqref{eq:annuity-ode} gives
\begin{equation*}
  \frac{d}{dT}\big(P(T)-P(T,y)\big)
  +x(T)\big(P(T)-P(T,y)\big)
  =yP(T,y),
\end{equation*}
with zero initial value.  The comparison principle gives $P(T,y)<P(T)$ for $y>0$ and the reverse inequality for $y<0$.

For the curvature statement, approximate the continuous fixed leg on a partition
\begin{equation*}
  0=t_0<t_1<\cdots<t_N=T
\end{equation*}
with positive accrual factors $\alpha_j=t_j-t_{j-1}$.  The weighted discrete annuity is
\begin{equation*}
  A_N(T)=\sum_{j=1}^{N}\alpha_jp(t_j).
\end{equation*}
The two-sided ratio-monotonicity and shifted-annuity-curvature arguments of Lemmas~\ref{lem:ratio-mono} and \ref{lem:annuity-curvature} use only positivity of the accrual weights.  Hence, for every partition, the shifted weighted-annuity graph is concave for $y>0$ and convex for $y<0$.  Passing to the limit as the mesh tends to zero gives the stated continuous curvature property.  The sign in \eqref{eq:continuous-shift-component} then follows exactly as in Theorem~\ref{thm:swap-shift-component}.
\end{proof}

The carry statement also has a direct continuous counterpart.  If the three points $(P(T_i),x(T_i))$ are convex, define
\begin{equation*}
  L_x^{\mathrm{cont}}
  =\lambda_1x(T_1)+\lambda_3x(T_3)-\lambda_2x(T_2)>0.
\end{equation*}
Over a short holding period $h$, the continuously accrued fixed-rate contribution is $hL_x^{\mathrm{cont}}$.  The common floating accrual cancels because $\lambda_1+\lambda_3=\lambda_2$.  Thus the continuous model has the same two components as the discrete market-swap construction: positive short-horizon carry under annuity-coordinate convexity and non-negative revaluation under an admissible instrument-based parallel shift.

The continuous formulation is not meant to replace market swap conventions.  It is a clean limiting model.  Actual par curves are quoted from swaps with specified payment frequencies, day-count fractions, calendars, floating-rate reset conventions, and collateral or discounting conventions.  The relevant horizontal coordinate remains the fixed-leg annuity: in the annual discrete model it is $P_n=\sum_{i=1}^{n}p_i$, in a market swap it is $A_n=\sum_{i=1}^{n}\alpha_i p_i$, and in the continuous model it is $P(T)=\int_0^T p(s)\,ds$.

\section{Empirical Evidence from U.S. Dollar Swap Curves}\label{sec:empirical}

The theoretical results identify curve shape in the relevant coordinate as a conditional carry and revaluation diagnostic, not as a model-independent arbitrage criterion.  This section examines how that diagnostic behaves in observed swap curves.  We proceed in three steps: we construct daily discount factors and fixed-leg annuities from quoted par rates; compare convexity classifications in annuity and maturity coordinates; and revalue each annuity-convex butterfly on the next available curve to measure departures from instrument parallelity and realized gross and transaction-cost-adjusted P\&L.  The pre-crisis single-curve period is the primary specification, while crisis, post-crisis, and full-sample results provide regime and robustness evidence.

\subsection{Data and curve construction}

We use the discontinued daily U.S. dollar swap-rate series distributed by the Federal Reserve Bank of St. Louis through FRED \cite{fredswap}.  The quoted maturities are
\begin{equation*}
  1,\ 2,\ 3,\ 4,\ 5,\ 7,\ 10,\ 30\ \text{years},
\end{equation*}
and the corresponding series are DSWP1, DSWP2, DSWP3, DSWP4, DSWP5, DSWP7, DSWP10, and DSWP30.  The observations are mid-market par swap rates against three-month LIBOR.  Requiring all eight quotes to be present leaves 4,078 daily curves from July 3, 2000 through October 28, 2016.

These series describe the historical LIBOR-based swap market rather than
the current SOFR-based market. They are used because FRED provides a
long, publicly accessible, and internally consistent daily panel across
the eight selected maturities. Although FRED publishes the overnight
SOFR fixing, it does not provide a comparably long historical panel of
daily SOFR overnight-indexed-swap par rates across these tenors
\cite{fredswap,fredsofr}. Current U.S. dollar SOFR swap-rate benchmarks
are available from market benchmark administrators
\cite{icesofrswap}, but they do not provide the same long public history
and tenor coverage as the FRED sample used here.

The empirical exercise should therefore be interpreted as a historical
study of the shape and realized performance of U.S. dollar swap curves
under a reproducible public-data specification. Applying the analysis to
a contemporary SOFR curve would replace the LIBOR swap quotations with
SOFR overnight-indexed-swap quotations and would incorporate the
corresponding overnight-compounding and market conventions. The
sequential bootstrap logic, the definition of the fixed-leg annuity, and
the annuity-coordinate butterfly construction would otherwise remain
the same.

For transparency, we use a deliberately simple single-curve bootstrap.  Fixed payments are semiannual with regular 30/360 accrual $\alpha=1/2$.  Quoted par rates are linearly interpolated to a semiannual maturity grid; the one-year quote is extended flat to six months.  If $S_n$ denotes the par rate at the $n$th semiannual date, the discount factors are obtained recursively from
\begin{equation}\label{eq:empirical-bootstrap}
  D_n=\frac{1-S_n\alpha\sum_{j=1}^{n-1}D_j}{1+S_n\alpha},
  \qquad
  A_n=\alpha\sum_{j=1}^{n}D_j.
\end{equation}
The maximum par-rate reconstruction error is $1.95\times10^{-12}$ basis points.  The single-curve convention is most natural in the pre-crisis LIBOR sample, which we treat as the primary specification.  Crisis, post-crisis, and full-sample results are reported as regime and robustness evidence; the post-crisis calculations should be viewed as a consistent single-curve exercise rather than a complete multi-curve valuation.

For each date and each triplet $T_1<T_2<T_3$, let $A_i$ and $x_i$ denote the quoted-tenor annuity and par rate.  We normalize the belly notional to one and define the annuity-coordinate wing weights
\begin{equation}\label{eq:empirical-annuity-weights}
  w_1=\frac{A_3-A_2}{A_3-A_1},
  \qquad
  w_3=\frac{A_2-A_1}{A_3-A_1}.
\end{equation}
They satisfy
\begin{equation*}
  w_1+w_3=1,
  \qquad
  w_1A_1+w_3A_3=A_2.
\end{equation*}
The normalized annuity-coordinate convexity spread is
\begin{equation}\label{eq:empirical-annuity-spread}
  B=w_1x_1+w_3x_3-x_2.
\end{equation}
For comparison, define the maturity-coordinate weights and spread by
\begin{equation*}
  \widetilde w_1=\frac{T_3-T_2}{T_3-T_1},
  \qquad
  \widetilde w_3=\frac{T_2-T_1}{T_3-T_1},
  \qquad
  \widetilde B=\widetilde w_1x_1+\widetilde w_3x_3-x_2.
\end{equation*}
A triplet is classified as annuity-convex when $B>0$ and maturity-convex when $\widetilde B>0$.  With eight quoted tenors there are $\binom{8}{3}=56$ triplets per date and 228,368 date--triplet observations.

\subsection{Observed curve shape}

Table~\ref{tab:regime-shape} summarizes the full sample and three market regimes.  Convex observations are more frequent in annuity coordinates than in maturity coordinates in every regime.  Across the full sample, 24.03\% of observations are annuity-convex, compared with 19.80\% under maturity weighting, and the classifications disagree in 4.23\% of observations.  Convexity is much less frequent during 2008--2009, when average curve concavity is also substantially larger.

\begin{table}[H]
\centering
\caption{Curve-shape frequencies by regime}
\label{tab:regime-shape}
\small
\begin{tabular}{lrrrrrr}
\toprule
Sample & Dates & Obs. & $B>0$ & $\widetilde B>0$ & Disagree & Mean $B$ \\
 & & & (\%) & (\%) & (\%) & (bp) \\
\midrule
Full sample & 4,078 & 228,368 & 24.03 & 19.80 & 4.23 & -18.80 \\
Pre-crisis, 2000--2007 & 1,873 & 104,888 & 26.25 & 21.38 & 4.88 & -16.92 \\
Crisis, 2008--2009 & 501 & 28,056 & 8.31 & 6.54 & 1.77 & -27.11 \\
Post-crisis, 2010--2016 & 1,704 & 95,424 & 26.22 & 21.98 & 4.24 & -18.42 \\
\bottomrule
\end{tabular}
\begin{minipage}{0.95\textwidth}
\footnotesize\emph{Notes:} $B>0$ and $\widetilde B>0$ are the percentages of date--triplet observations classified as convex in annuity and maturity coordinates, respectively.  ``Disagree'' is the percentage for which the convex/concave classifications differ.  $B$ is measured in basis points.
\end{minipage}
\end{table}

The location of convexity is highly nonuniform.  Table~\ref{tab:common-triplets-shape} and Figure~\ref{fig:common-triplets-shape} show results for six representative butterflies.  The $1$--$2$--$3$ triplet is annuity-convex on 60.27\% of dates and has a positive mean spread of 1.14 basis points.  Convexity declines sharply with maturity and is essentially absent for the $7$--$10$--$30$ triplet.  The use of annuity rather than maturity coordinates is most consequential in the intermediate curve: the $3$--$4$--$5$ classification changes on 9.98\% of dates.

\begin{table}[H]
\centering
\caption{Shape statistics for representative swap butterflies}
\label{tab:common-triplets-shape}
\small
\begin{tabular}{lrrrrr}
\toprule
Triplet & $B>0$ & Mean $B$ & Max $B$ & $\widetilde B>0$ & Disagree \\
 & (\%) & (bp) & (bp) & (\%) & (\%) \\
\midrule
1--2--3 & 60.27 & 1.14 & 24.17 & 55.96 & 4.56 \\
2--3--4 & 43.28 & -1.14 & 7.82 & 37.44 & 5.86 \\
3--4--5 & 35.70 & -1.34 & 10.61 & 25.72 & 9.98 \\
3--5--7 & 23.96 & -5.21 & 7.30 & 17.34 & 6.62 \\
5--7--10 & 9.76 & -5.78 & 1.48 & 3.95 & 5.81 \\
7--10--30 & 0.17 & -16.88 & 1.00 & 0.07 & 0.10 \\
\bottomrule
\end{tabular}
\end{table}

\begin{figure}[H]
\centering
\begin{tikzpicture}
\begin{axis}[
  ybar,
  width=0.94\textwidth,
  height=0.45\textwidth,
  bar width=7pt,
  ylabel={Convex observations (\%)},
  symbolic x coords={1--2--3,2--3--4,3--4--5,3--5--7,5--7--10,7--10--30},
  xtick=data,
  x tick label style={rotate=25,anchor=east},
  ymin=0,
  ymax=70,
  legend style={at={(0.5,1.02)},anchor=south,legend columns=2},
  enlarge x limits=0.08,
  grid=major
]
\addplot[fill=black!75] coordinates {
  (1--2--3,60.27) (2--3--4,43.28) (3--4--5,35.70)
  (3--5--7,23.96) (5--7--10,9.76) (7--10--30,0.17)};
\addplot[fill=black!25] coordinates {
  (1--2--3,55.96) (2--3--4,37.44) (3--4--5,25.72)
  (3--5--7,17.34) (5--7--10,3.95) (7--10--30,0.07)};
\legend{Annuity coordinate,Maturity coordinate}
\end{axis}
\end{tikzpicture}
\caption{Full-sample convexity frequency for representative triplets}
\label{fig:common-triplets-shape}
\end{figure}

These results are shape diagnostics rather than evidence of a model-independent arbitrage.  They show that observed curves do contain annuity-convex segments and that maturity-space screening can miss or reverse the sign of the theoretically relevant swap-butterfly spread.

\subsection{Next-observation parallelity and realized P\&L}

We next follow every annuity-convex observation to the next available complete curve.  Let $t_+$ denote that next observation date and let
\begin{equation*}
  h=\frac{t_+-t}{365}
\end{equation*}
be the elapsed Actual/365 Fixed year fraction.  Two observations occur on the final sample date and have no successor, leaving 54,878 realized cases.  The retained calendar gaps range from one to four days.  Fixed coupons retain the semiannual 30/360 convention, while the floating leg resets quarterly against three-month LIBOR.

For each original swap $i$, let $D_t(u)$ and $D_+(u)$ denote discount factors at dates $t$ and $t_+$ for horizons $u$ from the respective valuation dates.  The remaining fixed payment times are
\begin{equation*}
  \tau_{i,j}=\frac{j}{2}-h,
  \qquad j=1,\ldots,2T_i,
\end{equation*}
and the remaining fixed-leg annuity is
\begin{equation}\label{eq:remaining-fixed-annuity}
  A_{i,+}=\frac12\sum_{j=1}^{2T_i}D_+(\tau_{i,j}).
\end{equation}
Let $\delta_f=1/4$ denote the three-month floating reset interval.  Because the first floating coupon was fixed at initiation and every retained holding period satisfies $h<\delta_f$, the remaining floating-leg value is
\begin{equation}\label{eq:remaining-floating-leg}
  F_{i,+}
  =\frac{D_+(\delta_f-h)}{D_t(\delta_f)}-D_+(T_i-h).
\end{equation}
This identity follows by valuing the already-fixed first floating coupon
separately and telescoping the remaining projected floating coupons under
the single-curve convention. The quarter-year discount factor is obtained by log-linear interpolation of discount factors.  The par rate at $t_+$ of the original swap, with its shortened remaining schedule, is
\begin{equation}\label{eq:own-instrument-rate}
  K_{i,+}=\frac{F_{i,+}}{A_{i,+}}.
\end{equation}
We refer to $K_{i,+}$ as the swap's own-instrument par rate, distinguishing it from the quoted rate of a newly issued constant-maturity swap.  Let $K_i=x_i$ denote the original fixed rate.  A unit-notional receive-fixed position then has next-observation value
\begin{equation}\label{eq:receiver-next-value}
  V_{i,+}=A_{i,+}(K_i-K_{i,+}).
\end{equation}
The annuity-balanced butterfly value is therefore
\begin{equation}\label{eq:realized-butterfly-pnl}
  \Pi_+^{\mathrm{gross}}
  =w_1V_{1,+}-V_{2,+}+w_3V_{3,+}.
\end{equation}
We report P\&L in basis-point-value units per unit belly notional; one such unit equals $10^{-4}$ of notional.

Define the three own-instrument rate changes by
\begin{equation*}
  \Delta K_i=K_{i,+}-K_i.
\end{equation*}
Instrument parallelity requires $\Delta K_1=\Delta K_2=\Delta K_3$.  Our main departure measure is the range
\begin{equation}\label{eq:parallel-range}
  R_+=\max_i\Delta K_i-\min_i\Delta K_i.
\end{equation}
This also gives an exact level-versus-nonlevel decomposition.  Writing
\begin{equation*}
  \Delta K_i=\overline{\Delta K}+\varepsilon_i,
  \qquad \sum_i\varepsilon_i=0,
\end{equation*}
the realized P\&L decomposes as
\begin{equation}\label{eq:level-nonlevel-decomposition}
  \Pi_+^{\mathrm{gross}}
  =-\overline{\Delta K}\sum_{i=1}^{3}N_iA_{i,+}
   -\sum_{i=1}^{3}N_iA_{i,+}\varepsilon_i,
\end{equation}
where $(N_1,N_2,N_3)=(w_1,-1,w_3)$.  The first term is the realized level component.  The second contains slope and curvature movements.  Initial annuity balance makes the level exposure zero at inception, but the remaining annuities change slightly with roll-down; nonparallel residuals therefore remain the principal empirical risk.

For transaction costs, we impose a benchmark full bid--ask width of 0.25 basis points, interpreted as a 0.125-basis-point half-spread on each transaction.  The entry and exit costs are
\begin{align}
  \mathrm{TC}_{\mathrm{entry}}
  &=0.125\sum_{i=1}^{3}|N_i|A_i,\label{eq:entry-cost}\\
  \mathrm{TC}_{\mathrm{exit}}
  &=0.125\sum_{i=1}^{3}|N_i|A_{i,+}.\label{eq:exit-cost}
\end{align}
This is a transparent benchmark rather than an estimate of time-varying execution costs.

Table~\ref{tab:next-overall} gives the main results.  The median cross-instrument shift range is 1.98 basis points and the 90th percentile is 5.88 basis points.  Only 8.07\% of convex observations are parallel within 0.25 basis points, and only 9.95\% are within 0.50 basis points.  Gross P\&L is positive in 52.43\% of cases, with a mean of 0.137 and a median of 0.015 basis-point-value units.  After entry and exit costs, only 31.89\% of cases remain positive; the mean and median round-trip P\&L are $-1.593$ and $-1.538$ basis-point-value units.

\begin{table}[H]
\centering
\caption{Next-observation parallelity and realized P\&L}
\label{tab:next-overall}
\small
\begin{tabular}{lr}
\toprule
Statistic & Value \\
\midrule
Convex observations with a next curve & 54,878 \\
Median instrument-shift range (bp) & 1.978 \\
90th percentile instrument-shift range (bp) & 5.883 \\
Parallel within 0.10 bp (\%) & 4.698 \\
Parallel within 0.25 bp (\%) & 8.072 \\
Parallel within 0.50 bp (\%) & 9.955 \\
Gross P\&L positive (\%) & 52.434 \\
Positive after entry cost (\%) & 40.373 \\
Positive after round-trip cost (\%) & 31.887 \\
Mean gross P\&L (bp-value units) & 0.137 \\
Median gross P\&L (bp-value units) & 0.015 \\
Mean round-trip net P\&L (bp-value units) & -1.593 \\
Median round-trip net P\&L (bp-value units) & -1.538 \\
\bottomrule
\end{tabular}
\end{table}

Table~\ref{tab:next-triplets} and Figure~\ref{fig:next-triplets} report representative triplets.  Front-end convexity has a modestly favorable gross sign frequency, but transaction costs remove the advantage.  The $5$--$7$--$10$ triplet has the highest gross positive frequency among the commonly observed sectors, yet only 29.65\% of cases are positive after round-trip costs.  The $7$--$10$--$30$ results are based on only seven observations and should not be interpreted as a stable empirical pattern.

\begin{table}[H]
\centering
\caption{Next-observation results for representative convex butterflies}
\label{tab:next-triplets}
\small

\textit{Panel A: instrument parallelity}

\vspace{0.3em}
\begin{tabular}{lrrr}
\toprule
Triplet & Cases & Median $R_+$ (bp) & Parallel within 0.25 bp (\%) \\
\midrule
1--2--3 & 2,458 & 1.958 & 6.184 \\
2--3--4 & 1,764 & 1.096 & 15.079 \\
3--4--5 & 1,456 & 1.022 & 21.429 \\
3--5--7 & 977 & 1.070 & 18.833 \\
5--7--10 & 398 & 0.984 & 29.397 \\
7--10--30 & 7 & 1.050 & 0.000 \\
\bottomrule
\end{tabular}

\vspace{1em}
\textit{Panel B: realized P\&L}

\vspace{0.3em}
\begin{tabular}{lrrrrr}
\toprule
Triplet & Gross $>0$ & Entry net $>0$ & Round net $>0$ & Median gross & Median round net \\
 & (\%) & (\%) & (\%) & (bp-value) & (bp-value) \\
\midrule
1--2--3 & 55.370 & 42.555 & 31.855 & 0.046 & -0.927 \\
2--3--4 & 57.937 & 41.667 & 23.526 & 0.018 & -1.376 \\
3--4--5 & 55.426 & 40.797 & 23.420 & 0.006 & -1.787 \\
3--5--7 & 49.437 & 42.170 & 25.896 & -0.002 & -2.186 \\
5--7--10 & 68.593 & 43.970 & 29.648 & 0.014 & -2.846 \\
7--10--30 & 71.429 & 71.429 & 57.143 & 4.665 & 0.783 \\
\bottomrule
\end{tabular}
\end{table}

\begin{figure}[H]
\centering
\begin{tikzpicture}
\begin{axis}[
  ybar,
  width=0.96\textwidth,
  height=0.47\textwidth,
  bar width=5pt,
  ylabel={Positive P\&L observations (\%)},
  symbolic x coords={1--2--3,2--3--4,3--4--5,3--5--7,5--7--10,7--10--30},
  xtick=data,
  x tick label style={rotate=25,anchor=east},
  ymin=0,
  ymax=80,
  legend style={at={(0.5,1.02)},anchor=south,legend columns=3,font=\small},
  enlarge x limits=0.08,
  grid=major
]
\addplot[fill=black!80] coordinates {
  (1--2--3,55.37) (2--3--4,57.94) (3--4--5,55.43)
  (3--5--7,49.44) (5--7--10,68.59) (7--10--30,71.43)};
\addplot[fill=black!40] coordinates {
  (1--2--3,42.55) (2--3--4,41.67) (3--4--5,40.80)
  (3--5--7,42.17) (5--7--10,43.97) (7--10--30,71.43)};
\addplot[pattern=north east lines,pattern color=black] coordinates {
  (1--2--3,31.86) (2--3--4,23.53) (3--4--5,23.42)
  (3--5--7,25.90) (5--7--10,29.65) (7--10--30,57.14)};
\legend{Gross,After entry cost,After round-trip cost}
\end{axis}
\end{tikzpicture}
\caption{Positive next-observation P\&L frequency for convex observations.  The $7$--$10$--$30$ category contains only seven cases.}
\label{fig:next-triplets}
\end{figure}

\subsection{Empirical interpretation}

The empirical results support two distinct conclusions.  First, annuity coordinates are economically and statistically relevant.  They identify more convex segments than maturity coordinates and reverse the classification on a nontrivial fraction of dates, particularly in the $3$--$4$--$5$ and $3$--$5$--$7$ sectors.  This supports the paper's main geometric contribution for par-swap butterflies.

Second, the conditional shift result should not be interpreted as a stand-alone trading strategy.  Daily movements are rarely instrument-parallel, the median nonparallel dispersion is large relative to the initial convexity carry, and gross next-observation P\&L is only slightly more likely to be positive than negative.  A 0.25-basis-point full bid--ask benchmark turns the mean and median round-trip result decisively negative.  The instrument-based theorem therefore isolates a favorable level-shift and carry component, while actual performance is governed by roll-down, changes in annuities, slope and curvature movements, and execution costs.

\section{Conclusion}

This paper studies the shape of zero and par yield curves through a unified carry-and-convexity analysis of zero-bond and par-swap butterflies under instrument-based parallel shifts.  The convention fixes the component instruments and applies the same additive change to each instrument's own quote.  It is a transparent desk-level scenario and not a complete law for the future rolled curve.

For zero-coupon bonds, maturity is the natural coordinate.  A zero-cost, duration-neutral butterfly has a non-negative instantaneous revaluation under a common yield shift.  Under the holding-period rule $y_h(T_i-h)=y_0(T_i)+a$, convexity of the selected initial zero yields adds positive benchmark carry.  Under a general future curve, however, roll-down and non-level yield changes can offset that component.

For par swaps, the fixed-leg annuity or PV01 is the correct coordinate.  The nonlinear bootstrap relation implies that shifted annuities are concave functions of initial annuities after a positive common par-rate shift and convex functions after a negative shift.  This curvature reversal, combined with the sign of the rate change, makes the instantaneous value of the annuity-weighted butterfly non-negative in both directions.  Initial annuity-coordinate convexity enters through short-horizon carry because the common floating accrual cancels across the zero-net-notional butterfly.

The empirical evidence confirms the practical importance of the coordinate choice.  In 4,078 daily U.S. dollar swap curves, annuity and maturity weighting disagree on 4.23\% of all triplets and on nearly 10\% of $3$--$4$--$5$ observations.  Convexity is concentrated at the front end.  At the same time, the direct next-observation test shows that actual swap-rate changes are seldom instrument-parallel.  Gross P\&L for initially convex butterflies is positive in 52.43\% of cases, but only 31.89\% remain positive after entry and exit costs based on a full quoted bid--ask width of 0.25 basis points, and the average net result is negative.

The resulting interpretation is deliberately limited and economically useful.  An annuity-convex segment identifies positive benchmark carry and non-negative revaluation under the specified instrument-based level-shift family.  It does not eliminate slope, curvature, roll-down, basis, or transaction-cost risk.  The theory therefore supplies a curve-construction and trade-design diagnostic, while the empirical results quantify why conditional convexity is not equivalent to realized arbitrage.

The common-shift direction is also the one used in standard aggregate DV01 and PV01 calculations, so the conditional geometry is relevant to portfolio risk even when richer key-rate measures are required.  Over very short horizons, instrument-based and rolled-curve parallel scenarios differ only by a first-order roll-down term, but the instrument-based formulation remains the appropriate position-relative convention for following an existing butterfly.  These observations explain both the practical usefulness and the deliberate limits of the framework.

\end{document}